\documentclass[12pt]{article}

\usepackage[margin=1in]{geometry} 

\usepackage{setspace} 

\usepackage{xr}
\makeatletter
\newcommand*{\addFileDependency}[1]{
  \typeout{(#1)}
  \@addtofilelist{#1}
  \IfFileExists{#1}{}{\typeout{No file #1.}}
}
\makeatother

\usepackage{graphicx,verbatim,array,multicol, xcolor, lscape, mathrsfs, color,enumitem}
\usepackage{psfrag, amsmath, amsfonts, epsfig, fancybox, setspace,soul,amsthm,natbib,algorithm, algorithmicx, algpseudocode, float, subcaption, booktabs, pdflscape}
\usepackage[hidelinks]{hyperref}
\def\bW{{\bf W}}
\def\bZ{{\bf Z}}
\def\bL{{\bf L}}
\def\bw{{\bf w}}
\def\supone{^{(1)}}
\def\supzero{^{(0)}}
\def\supg{^{(g)}}
\def\transpose{{\sf \scriptscriptstyle{T}}}
\def\bU{{\bf U}}
\def\bUhat{{\widehat \bU}}
\def\trans{^{\transpose}}
\def\var{\text{Var}}
\def\bbeta{\boldsymbol{\beta}}
\def\bomega{\boldsymbol{\omega}}
\def\bPi{\boldsymbol{\Pi}}
\def\calt{\mathcal{T}}
\newtheorem{theorem}{Theorem}
\setcitestyle{authoryear}
\def\bSigma{\boldsymbol{\Sigma}}
\def\bmu{\boldsymbol{\mu}}
\newtheorem{assumption}{Assumption}

\newcommand{\norm}[1]{\left \lVert #1 \right \rVert}

\begin{document}

\setlength\parindent{15pt} 

\thispagestyle{empty}

\begin{center}
\LARGE{A Functional-Class Meta-Analytic Framework for Quantifying Surrogate Resilience}\\

\vspace*{15mm}
\normalsize Emily Hsiao$^{1}$ and Layla Parast$^{2}$ \\
\vspace*{10mm}
\normalsize $^{1}$Department of Statistics and Data Sciences, University of Texas at Austin, \\ \normalsize 105 E 24th St D9800, Austin, TX 78705 \\
\normalsize $^{2}$Department of Statistics and Data Sciences, University of Texas at Austin, \\ \normalsize 105 E 24th St D9800, Austin, TX 78705, parast@austin.utexas.edu \\

\end{center}
\clearpage
\thispagestyle{empty}
\begin{abstract}
 A surrogate marker is a biomarker or other physical measurement used to replace a primary outcome in clinical trials to evaluate a treatment effect when the primary outcome of interest is costly, invasive, or takes a long time to observe. However, replacing a primary outcome with a surrogate can lead to the ``surrogate paradox,” in which a treatment appears beneficial based on the surrogate but is actually harmful with respect to the primary outcome. In this paper, we propose a functional class-based method to assess resilience to the surrogate paradox in a meta-analytic setting. Our method leverages data from $K$ completed studies in which the surrogate marker and primary outcome have been measured to make inference on a new study in which only the surrogate is measured. We do not assume direct transportability of the conditional mean function from the completed studies to the new study; instead, we consider deviations of functions from those observed in the completed studies to estimate the ``resilience probability" i.e., the probability of the surrogate paradox in the new study. We investigate the performance of our proposed method through a simulation study and apply our method to data from clinical trials in schizophrenia.
\end{abstract}

\noindent \textbf{Keywords:}  causal inference, clinical trial, meta-analysis, sensitivity, surrogate paradox, treatment effect

\clearpage
\setcounter{page}{1}
\section{Introduction}
A surrogate marker is a biomarker or other physical measurement used to replace a primary outcome in clinical trials to evaluate a treatment effect when the primary outcome of interest is costly, invasive, or takes a long time to observe \citep{temple1999surrogate}. The Food and Drug Administration in the United States currently allows drugs to be approved early through the Accelerated Approval Program, if they have proven effectiveness on a known surrogate marker \citep{FDA_program}. There are many examples in which this program has been beneficial in allowing effective treatments to reach patients sooner. For example, Gleevec (imatinib), a treatment for individuals with chronic myelogenous leukemia, was initially approved through this program based on a demonstrated effect on surrogate markers including hematologic and cytogenetic responses, i.e., various blood and platelet counts \citep{cohen2002approval,cohen2002gleevec}. This early approval was followed by subsequent long-term studies that demonstrated a beneficial treatment effect on progression-free and overall survival \citep{cohen2005us,druker2006five,hochhaus2017long}. However, the use of surrogate markers remains controversial due to the phenomenon known as the surrogate paradox, in which a treatment has a beneficial effect on a surrogate, yet later is found to actually be harmful with respect to the primary outcome of interest \citep{chen2007criteria, vanderweele2013surrogate}. The most famous case of the surrogate paradox involved anti-arrhythmia drugs meant to protect against cardiovascular death which was initially approved when studies found that it lowered rates of ventricular arrhythmia \citep{fleming1996surrogate,temple1999surrogate,lesko2001use}. It was later found to in fact increase the risk of cardiovascular death, but the discovery was made only after the mistake had cost thousands of lives \citep{echt1991mortality,moore}. Thus, when evaluating a surrogate marker and using a surrogate marker in a trial, it is highly important to properly assess the risk of the surrogate paradox. 

Notably, in this paper, we specifically focus on the surrogate paradox. Though there are many existing surrogate validation techniques, it is not our primary goal to address these different techniques, as the surrogate paradox may occur regardless of the chosen validation method \citep{elliott2023surrogate,stijven2024proportion,gilbert2008evaluating,huang2011comparing,freedman1992statistical,parast2015robust,parast2024rank,burzykowski2005evaluation,buyse2000validation,buyse1998criteria}. The surrogate paradox has been well-discussed in the existing surrogate literature \citep{chen2007criteria, vanderweele2013surrogate,yin2020novel,wu2011sufficient}. Statistical papers that propose surrogate validation methods typically impose formal conditions that, if they hold, ensure the paradox will not occur (our own included). \citet{wang2002measure} showed that the paradox will not occur if the functions governing the relationship between the surrogate and the primary outcome satisfy three sufficient but not necessary conditions. Recent work by \citet{hsiao2025} proposed empirical nonparametric methods to formally test these conditions with the observed data.  In the specific context of one-parameter exponential families, generalized linear models and Cox's proportional hazard models, \citet{wu2011sufficient} derived sufficient conditions to prevent the surrogate paradox.  \citet{yin2020novel} further investigated the conditions that would ensure protection against the paradox, but only in the case when both surrogate and primary outcome are binary. While useful, since these conditions are sufficient (not necessary), researchers still lack practical guidance for interpreting results if they determine or suspect that one or more of these conditions do not hold. For example, if one condition does not hold, is the surrogate then invalid? 

Furthermore, whether the conditions hold in the study used to validate the surrogate is, arguably, irrelevant. That is, surrogate validation methods are generally applied to data from a completed study (or studies) where both the surrogate and the primary outcome have been measured. Thus, we know exactly whether the surrogate paradox occurs in this completed study as we have both the estimated treatment effect on the surrogate and on the outcome. The question is whether the surrogate paradox will occur in a new study if that new study actually uses the surrogate alone to make inference on the treatment effect. Most relevant to this question is work by \citet{elliott2015surrogacy} which proposed methods to assess the probability of the surrogate paradox in a new study. Their proposed method uses a meta-analytic causal association framework involving multivariate normal trial-
level random effects and additionally allows one to determine the required size of a beneficial effect of the treatment on the surrogate that minimizes the risk
of a harmful effect of the treatment on the outcome. \citet{shafie2023incorporating} expanded on this approach by incorporating baseline covariates, allowing the covariate effects on the surrogate and outcome to vary across trials. Together, these methods provide powerful tools for quantifying the risk of the surrogate paradox, but they rely on strong modeling assumptions, including additive linearity and multivariate normality of the joint distribution of surrogates and outcomes in both treatment groups. In this paper, we build upon this framework by allowing the conditional distributions that link the surrogate and outcome to come from a broad functional class, drawing on flexible representations that capture complex trial-level deviations while incorporating knowledge about their structure that is learned from the completed studies.

The problem that we aim to address in this paper is described as follows. Assume that we have a surrogate marker which has been validated using data from a number of previous studies. In addition, suppose we have data from $K$ completed studies in which  both the surrogate and the primary outcome have been measured. We now have a new study, wherein only the surrogate has been measured, and we find that the treatment has a positive (beneficial) effect on the surrogate. \textit{How can we assess the likelihood that the actual treatment effect on the primary outcome in this new study will also be positive, without ever measuring the primary outcome in this new study?} In this paper, we answer this question by proposing a functional class-based meta-analytic method that allows us estimate the probability that the true treatment effect will be negative. In an ideal setting, we would want this probability, referred  to as the ``resilience probability" \citep{hsiao2025resilience}, to be small. 

Previous work in \citet{hsiao2025resilience} has proposed a method specific to the setting where there is only a single prior study. While useful, a significant limitation is the requirement that the researcher specify fixed variance parameters which are typically unknown. In contrast, in this paper, we focus on a setting where there are $K>1$ prior studies, and thus, overcome this limitation by estimating the variance parameters via a meta-analytic approach. 

\vspace*{-0.2in}
\section{Methods}
\vspace*{-0.1in}
\subsection{Notation and Meta-Analytic Setting}

Let $G$ be a binary randomized treatment indicator, where $G = 0$ in the control group and $G = 1$ in the treatment group. Let $Y$ denote the primary outcome of interest, and $S$ the value of the surrogate marker. Without loss of generality, we assume higher values of $Y$ and $S$ are ``better". In our setting, we have $K$ completed studies where both $S$ and $Y$ have been measured, and a new study, referred to as Study $K+1$ where only $S$ has been measured. We use potential outcomes notation for $Y$ and $S$ such that $Y_{k}\supg$ and $S_{k}\supg$ denote the primary outcome and surrogate under treatment $g$ in study $k$. Of course, only one of the potential outcomes, e.g., $S_k\supone$ or $S_k\supzero$, can be observed for each individual.  For each study $k = 1,...,K$, the observed data consists of $\{S_{igk}, Y_{igk}\}_{i=1}^{n_{k,g}}$,  where $n_{k,g}$ is the sample size in group $g$ of study $k$ and $n_{k}$ is the total sample size in study $k$. For $k=K+1$, the data consists of only the surrogate measured in both groups i.e., $\{S_{i0(K+1)}\}_{i=1}^{n_{K+1,0}}$ from the control group and $\{S_{i1(K+1)}\}_{i=1}^{n_{k+1,1}}$ from the treatment group. To ease notation, we generally work with the study-level vectors $Y_{kg} =  (Y_{1gk},...,Y_{(n_{k,g})gk})'$ and $S_{kg} =  (S_{1gk},...,S_{(n_{k,g})gk})'$ for treatment group $g$ in study $k$

Our ultimate quantity of interest is the treatment effect on the primary outcome in Study $K+1$: 
$$\Delta_{K+1} = E\left (Y_{K+1}\supone - Y_{K+1}\supzero \right ) $$
Specifically, we aim to develop a method to estimate the probability that this treatment effect is negative: $$p_0\equiv P(\Delta_{K+1}<0),$$ where this quantity, in the $K=1$ setting, was proposed in \citet{hsiao2025resilience}.  Since $Y$ is not measured in Study $K+1$, we aim to use information from the $K$ completed studies and the surrogates which are measured in Study $K+1$ to estimate this quantity.

Existing work has offered methods to calculate $\Delta_{K+1}$ and explicitly test the null hypothesis that $\Delta_{K+1} = 0$ \citep{wang2020model,price2018estimation,parast2022using}.  However, these methods rely on strict assumptions on transportability from the completed studies to the new study that are unlikely to hold in practice. Most often, the methods require estimation of the group-specific conditional mean functions defined as $\mu_{k,g}(s) = E(Y\supg_{k} | S\supg_{k}=s)$, i.e., the conditional mean of the primary outcome given the surrogate in group $g$. While these conditional mean functions can be flexibly estimated in the completed study either parametrically or nonparametrically \citep{freedman1992statistical,wang2020model,parast2016robust}, assuming that the function remains the same in the new study is likely unrealistic. There are many factors which can influence the relationship between a surrogate and a primary outcome from one study to another, including for example, differences in the administration of the treatment and/or in the study population. In the following section, we propose an approach that relaxes such a strict transportability assumption and allows for the conditional mean functions to differ from the completed studies.

\subsection{Proposed Functional Class-Based Method}
\subsubsection{Functional Class}
The key feature of our proposed method is that we assume that the true underlying conditional mean function in the treatment and control groups, each belong to a flexible functional class $F(\mu, \Sigma)$ characterized by a mean vector $\mu$ and covariance matrix $\Sigma$. For instance, $F$ could represent a Gaussian–process family, a Fourier–series representation,
or a polynomial functional class, each of which can flexibly approximate a broad range of
well-behaved distributions \citep{rudinprinciples}. In this paper, we specifically consider the Gaussian Process (GP) class. We denote the true conditional mean function in group $g$ as  $m_g(s)$ and we assume that, for each study, the observed conditional mean function is a random deviation from $m_g(s)$. Formally, we posit the following:

$$Y_{kg} = F_{k,g}+ \epsilon, ~~\mbox{where}$$
$$F_{k,g} \sim GP(m_g(S_{kg}), C(S_{kg}, S_{kg};\sigma^2_g, \theta_g))$$
$$\epsilon \sim N(0, diag(v^2_g)),$$
and GP refers to a Gaussian Process with mean vector $m_g(S_{kg})$ and a covariance kernel $C(\cdot)$ defined on the values of the vector $S_{kg}$, parameterized by $\sigma_g^2$ and $\theta_g$. In essence, we assume that the observed $Y$ vector of values in each study is an independent draw from a Gaussian Process with the same mean function and covariance kernel parameters. Notice that by assuming a Gaussian Process, each  $Y_{kg}$ vector is multivariate normal with mean $\mu = m_g(S_{kg})$ and covariance matrix $\Sigma = C(S_{kg}, S_{kg}; \sigma_g^2, \theta_g) + diag(v_g^2)$.

\subsubsection{Parametrization and Estimation}
Next, we address the question of how to parameterize and estimate $m_g$ and $C(\cdot)$. If we only had one completed study, our choices here would be limited. We would be necessarily forced to simply take $m_g$ to be the estimated mean in the completed study, for example using nonparametric estimation; the covariance kernel would be more challenging since with only a single study, direct estimation of the full variability is not feasible. Fortunately, in our meta-analytic setting we may leverage the $K$ studies for estimation. Consider a model that can be parameterized by a vector $\beta$, such as a linear model or a polynomial basis, as long as the estimation of $\beta$ remains tractable. We propose to estimate $m_g$ specifically using cubic B-splines given its flexibility while maintaining reasonable tractability; that is,
\begin{equation}
m_g(S_{kg}) = \beta'_g B(S_{kg})
    \label{m_model}
\end{equation}
where $\beta_g$ denotes the unknown parameter vector and $B(\cdot)$ denotes the spline basis matrix from the input vector with specified knots and boundary knots; for the covariance kernel of the Gaussian Process, we use the radial basis function (RBF) kernel. Given the model defined above, it follows that the likelihood of the vectors $Y_{1g},...,Y_{kg}$ for each treatment group takes the form:
\begin{align*}
    p(Y_{1g},..,Y_{kKg}; \beta_g, \sigma^2_g, \theta_g, v^2_g) &= \prod_{k=1}^K p(Y_{kg}; S_{kg}, \beta\supg, \sigma^2_g, \theta_g, v^2_g) \\
    &= \prod (2\pi)^{n_k} |C_k|^{-1/2} \exp\left(\frac{1}{2}(Y_k - m_g(S_{kg}))' C_k^{-1} (Y_k - m_g(S_{kg}))\right)
\end{align*}
where $|C_k|$ denotes the determinant of $C_k$ and $C_k = RBF(S_{kg}, S_{kg};\sigma^2_g, \theta_g) + v^2_gI$. The log likelihood becomes 
$$\sum_{k=1}^K -\frac{1}{2} \log(|C_k|) -\frac{1}{2}(Y_{kg} - m_g(S_{kg}))'C_k^{-1}(Y_{kg} - m_g(S_{kg})) + Const.$$
We maximize this likelihood for the treatment and control groups separately to get the estimated vector of parameters $(\widehat{\beta}_0,\widehat{\beta}_1, \widehat{\sigma}^2_0, \widehat{\sigma}^2_1, \widehat{\theta}_0, \widehat{\theta}_1)$ using any optimization method (in our implementation, we use \texttt{nlminb}). Notably, this finite–dimensional representation of the mean within our specification of $m_g(\cdot)$ can be viewed as a truncated Karhunen–Loève 
(or other orthogonal) expansion of a potentially infinite-dimensional mean function 
\citep{adler2009random}, but we adopt the finite basis form as a modeling assumption.

\subsubsection{Resilience Probability}
To obtain an estimate of the resilience probability for Study $K+1$,  we use our estimated class to generate a synthetic distribution of potential $\Delta_{K+1}$ values. Specifically, we plug the estimated parameters into our posited distribution, along with the observed $S_{(K+1)g}$ vectors in the new study, and then generate a set of synthetic $Y_{(K+1)0}$ and $Y_{(K+1)1}$ vectors. That is, we calculate the estimated mean vector and covariance matrix for the new study as:
$$\widehat{m}_g(S_{(K+1)g}) = \widehat{\beta}_g \cdot B(S_{(K+1)g})$$
$$\widehat{\Sigma}_g = C(S_{(K+1)g}, S_{(K+1)g}; \widehat{\sigma}_g^2, \widehat{\theta}_g) + diag(\widehat{v}_g^2),$$
and sample from a multivariate normal with parameters $\widehat{m}_g(S_{(K+1)g})$ and $\widehat{\Sigma}_g$ to obtain the synthetic $Y$ data for Study $K+1$. This idea is illustrated in Figure \ref{fig:idea} where the black line reflects $\widehat{m}_g(S_{(K+1)g})$ and each grey line is randomly generated from a multivariate normal with parameters $\widehat{m}_g(S_{(K+1)g})$ and $\widehat{\Sigma}_g$. Each grey line produces a resulting synthetic set for $Y_{(K+1)g}$. We use this synthetic data in each group to obtain a corresponding synthetic $\Delta_{K+1}$. Repeating this process $J$ times (e.g., $J=500$) yields an empirical distribution $ \left\{\widehat{\Delta}_{K+1,j} \mid j=1, \cdots J \right\}.$ Finally, we estimate the resilience probability, as:
\begin{equation*}
    \widehat{p} = J^{-1} \sum_{j-1}^J I(\widehat{\Delta}_{K+1, j} <0).
\end{equation*}
A small value of $\widehat{p}$ would reflect a low probability of the surrogate paradox in Study $K+1$, whereas a large value would indicate that the surrogate paradox is likely in Study $K+1$.


\section{Properties and Inference \label{theory}}
\subsection{Asymptotic Properties}
To understand the behavior of $\widehat{p}$, we next examine its theoretical properties. After stating the assumptions that justify the construction of the estimator, we show that our proposed estimator converges to the target quantity $p_0$. We then characterize its variance and propose a procedure to perform inference in practice. Specifically, we impose the following assumptions:

\begin{assumption}[Independent Studies]\label{asmp:ind}
For each study $k = 1,\dots,K$, we assume (i) $\{(S_{k0},$ $S_{k1},Y_{k0},Y_{k1})\}_{k=1}^K$ are mutually independent across studies; (ii) the observed data in groups $g=0$ and $g=1$ are independent; (iii) the per-study sample sizes are bounded i.e., $n_k \le C < \infty$ for all $k$.
\end{assumption}

\begin{assumption}[IID, Study $K+1$]\label{asmp:K1}
Within Study $K+1$, the observed surrogates in the control group, $\{S_{i(K+1)0}\}_{i=1}^{n_{(K+1),0}}$, are independent and identically distributed (i.i.d.), and the observed surrogates in the treatment group, $\{S_{i(K+1)1}\}_{i=1}^{n_{(K+1),1}}$, are i.i.d. Furthermore, the sample proportion satisfies $n_{(K+1),0}/n_{K+1}  \rightarrow r \in (0,1)$ as $n_{(K+1),0}, n_{(K+1),1} \to \infty$.
\end{assumption}

\begin{assumption}[Functional Class and Basis]\label{asmp:FC}
For each study $k$ and group $g\in\{0,1\}$,
$$
Y_{kg} = F_{kg} + \varepsilon_{kg},\qquad F_{kg}\sim GP\big(m_g(\cdot),\,C(\cdot,\cdot;\sigma_g^2,\theta_g)\big),
$$
with $m_g(s)=\beta_g' B(s)=\sum_{l=1}^L\beta_{gl}B_l(s)$ and known basis $B(s)=(B_1(s),\ldots,B_L(s))^\top$ such that each $B_l(\cdot)$ is bounded and continuous on the support of $S$, and $L$ is fixed (does not grow with $n$). The noise satisfies $\varepsilon_{kg}\sim N(0,v_g^2)$ independent across $k$ and independent of $F_{kg}$ given $S_{kg}$.
\end{assumption}

\begin{assumption}[Parameter and Identifiability Regularity]\label{asmp:regularity}
(i) The parameter spaces for $\beta_g$, $\sigma_g^2$, and $\theta_g$ are compact and the true parameter lies in the interior; (ii) The design matrix for the basis in each study has full column rank (no perfect multicollinearity), so $\beta_g$ is identified; 
(iii) finite moments: $\mathbb{E}[\|Y_{kg}\|_2^2]<\infty$ and $\|B(S_{kg})\|_2\le B^*<\infty$ for all $k$; (iv) $\lim_{K,n_s\to\infty} \kappa$ exists and is finite, where $\kappa = (n_s/K)$ and $n_s = [n_{(K+1),1}n_{(K+1),0}]/n_{K+1}$; (v) standard MLE regularity conditions hold for the GP likelihood (smooth log-likelihood, identifiable parameters, non-singular Fisher information).
\end{assumption}

\begin{assumption}[Covariance Regularity Conditions]\label{asmp:covregularity}
Let $\phi_g = (\sigma_g^2, \theta_g)$ and denote $C(S_{kg},S_{kg};$ $\sigma_g^2,\theta_g)$  as  $C_k(\phi_g)$ for notational simplicity. The covariance matrix of each study $C_k(\phi_g)$ satisfies four conditions (which are satisfied by the RBF kernel under compact $\phi_g$ space): (i) $C_k(\phi_g)$ is positive definite for all $\theta$; (ii) the eigenvalues of $C_k(\phi_g)$ are bounded above and below by $\lambda_{min} > 0$ and $\lambda_{max} < \infty$; (iii) $C_k(\phi_g)$ is Lipschitz continuous in $\phi_g$ with respect to the Frobenius norm; (iv) $C_k(\phi_g)$ is continuously differentiable in $\phi_g$.
 \end{assumption}

 \begin{theorem}[Asymptotic Properties of $\widehat{p}$]\label{thm:asymp_normality}
Under Assumptions \ref{asmp:ind}--\ref{asmp:covregularity}, the estimator $\widehat{p}$ is consistent for $p_0$, and satisfies
$$
\sqrt{\frac{n_{(K+1),1}n_{(K+1),0}}{n_{K+1}}}(\widehat{p} - p_0) \xrightarrow{d} N(0, \sigma_P^2),$$
as $K\to\infty$ and $n_{(K+1),0}, n_{(K+1),1} \to \infty$, where $\sigma_P^2$ is the asymptotic variance defined in Appendix A.
\end{theorem}
\noindent 
A detailed proof is included in Appendix A. At a high-level, the proof involves first expressing $p_0$ as a functional of the unknown parameters $\Gamma_g =(\beta_g, \sigma_g, \theta_g)'$, the basis components, and the true cumulative distribution functions (CDFs) of $S_{K+1}\supzero$ and  $S_{K+1}\supone$. We then prove consistent estimation and asymptotic normality of the estimators $\widehat{\Gamma}_g =(\widehat\beta_g, \widehat\sigma_g, \widehat\theta_g)'$ obtained using data from the $K$ studies via our proposed procedure. These asymptotic results, combined with the smoothness of the functional defining $p_0$, allow us  to establish consistency and asymptotic normality of $\widehat{p}$. 

\subsection{Uncertainty Quantification of $\widehat{p}$}
 To quantify uncertainty of our estimator, we focus on standard error estimation and confidence interval construction. Given the complex, but tractable, form of  $\sigma_P$ given in Appendix A, we offer two methods. The first is a fully nonparametric bootstrap which fully bootstraps the entire procedure and is, as expected, computationally intensive. Specifically, for a large number of iterations $R$, we resample $K$ studies with replacement and calculate a new estimate of $\widehat{\Gamma}_r$. We then resample $S_B$ (from Study $K+1$) with replacement, and calculate the bootstrap estimate of the resilience probability, denoted $\widehat{p}_r$. Finally, we take the empirical standard deviation of the $\widehat{p}_r$ estimates as an estimate of the standard error. We construct a 95\% confidence interval as the 2.5th and 97.5th percentiles of the bootstrap estimates. This method leverages the nonparametric bootstrap to target both the parameter uncertainty and variation arising from the randomness in $S_B$ at once. 

As an alternative to the fully nonparametric bootstrap, we propose a Partially Analytic Bootstrap (PAB) procedure. In Appendix A we show that the variance of the estimator can be expressed as the sum of two components, the first of which involves the variance from the MLE-estimated parameters and the second involves the randomness from the $S_B$ values. We propose to approximate the first component using the asymptotic covariance matrix of the MLE of $\Gamma_g$ and a gradient term evaluated at $\widehat{\Gamma}_g$. That is, we obtain the covariance matrix by calculating the numerical Hessian of the log-likelihood at $\widehat{\Gamma}_g$ using the \texttt{numDeriv} package and inverting it. We obtain the gradient term by way of numerical gradients with the \texttt{numDeriv} package. In essence, this gradient is calculating the approximate change in $\widehat{p}$ as $\widehat{\Gamma}_g$ changes. Applying the delta method, we obtain an approximation of the parameter uncertainty. For the second component, which involves the randomness from the $S_B$ values, we bootstrap the observed $S_B$. Specifically, we generate $R$ (for large $R$) independent samples of $S_B\supzero$ and $S_B\supone$ with replacement, calculate the value of $\widehat{p}_r$ which would have been obtained with the estimated $\widehat{\Gamma}$, and then calculate the empirical variance of the estimated $\{\widehat{p}_r, r=1,...,R\}$. Together, these two components provide an estimate of the standard error of $\widehat{p}$. We construct a 95\% confidence interval for $p_0$ using a normal approximation with this standard error estimate. 

As the fully nonparametric bootstrap is computationally intensive, this proposed PAB procedure offers significant computational advantages. However, there are four important disadvantages. First, the PAB inherently relies on correct model specification for accurate standard error estimation. Second, it relies on several approximations within the first component, namely, the MLE normality approximation, the two numerical gradient approximations, and the first-order Taylor expansion. The result is that if sufficient conditions are not met, this approach may yield slightly inaccurate results.  Third, the confidence interval relies on the asymptotic normal approximation, rather than percentiles as the fully nonparametric bootstrap utilizes. Lastly, it is prone to numerical errors, especially with the inversion of the information matrix. For implementation in our numerical studies, when the PAB method fails due to numerical errors, the nonparametric bootstrap method is used as a fallback for that particular iteration. In our numerical studies in Section \ref{sims}, we investigate and compare the finite sample performance of the fully nonparametric bootstrap and the PAB method procedure in terms of  accurate approximation of the empirical variance and coverage of the constructed confidence interval. 

\vspace*{-.2in}
\section{Simulation Study}\label{sims}
We evaluated our proposed method across six data generation settings with two configurations:  $K=10$ studies with $n=100$ per treatment group per study, and $K=25$ studies with $n=10$  per treatment group per study (the latter mirroring our data application). Settings~1-3 
were designed to assess whether the method correctly identifies high risk when the true probability of the 
paradox in Study $K+1$ is substantial. Settings 4-6 were designed to assess whether our proposed method can correctly result in a low estimated resilience probability when the probability of the paradox is indeed low.   For example, in the $K=10$ and $n=100$ sample size setting, the true paradox probabilities in Settings~1-3 were 0.586, 0.770, and 0.507, respectively, whereas the true paradox probabilities in Settings~4-6 were 0.029, 0.078, and 0, respectively. The  settings varied in function complexity (linear, polynomial, and trigonometric) and degree of  distributional shift between the completed studies and Study $K+1$. Detailed specifications of the mean functions, distributions, and 
Gaussian Process parameters are provided in Appendix B. 

Figure \ref{fig:high} visualizes the completed studies in Settings 1-3 when $K=10$ and the sample size per study and group is 100; for clarity, this figure shows data from 5 of the 10 completed studies with each study in a different color. In these figures, the black line represents the true mean function and the red line reflects the cubic B-splines estimated mean function based on the completed studies, as described above. Setting 1 is meant to be a simple linear model in which the surrogate paradox occurs because the distribution of the surrogate in the new study ($K+1$) is different, as described next. In this setting, in the prior studies, $S\supzero \sim N(3, 3)$ and $S\supone \sim N(4, 3)$ (top panel of Figure \ref{fig:high}), whereas in Study $K+1$, data were generated such that $S\supzero \sim N(4.75, 1)$ and $S\supone \sim N(5.25, 1)$.  In Setting 2,  in the completed studies, $S\supzero \sim N(0.9, 1.5)$ and $S\supone \sim N(2.2, 4.5)$ (middle panel of Figure \ref{fig:high}), whereas in Study $K+1$, data were generated such that $S\supzero \sim N(-0.7, 1)$ and $S\supone \sim N(-0.2, 2)$. Here, the control group mean function is non-linear; the spline estimation successfully captures this polynomial structure, as reflected by the close alignment between the red and black lines.  In Setting 3,  in the completed studies  $S\supzero \sim N(5,1)$ and $S\supone \sim N(6,2)$ (bottom panel of Figure \ref{fig:high}), whereas in Study $K+1$, data were generated such that  $S\supzero \sim N(4.1, 0.5)$ and $S\supone \sim N(4.1, 0.5)$. In this setting, the true mean function is non-polynomial with a higher signal-to-noise ratio than previous settings; nevertheless, the spline estimation approximates it reasonably well. Figure \ref{fig:low} visualizes the completed studies in Settings 4-6 when $K=10$ and the sample size per study and group is 100, similar to Figure \ref{fig:high}; see Appendix B for data generation details for these settings.

 In each iteration, we generated data and used our proposed method using linear, cubic, and cubic splines for the specified mean function to obtain $\widehat{p}$; results were summarized over 1000 iterations per simulation setting. We also compared our results with those from the maximum likelihood approach of \citet{elliott2015surrogacy}, detailed in Appendix B. Because the assumptions required by their method are not fully met in our simulation settings, this comparison is provided for illustrative purposes and should not be interpreted as a critique of their approach.
 
 Simulation results for all settings where $K=10$ and $n=100$, displayed in the top portion of Table \ref{sim_results}, show that the proposed method using cubic splines to estimate the resilience probability result in estimates close to the true probability of the surrogate paradox, both when this probability is high (Settings 1-3) and when this probability is low (Settings 4-6). The proposed method using a linear or cubic mean function performs well when the specification is correct, but poorly otherwise, as expected. The standard error estimates obtained from the fully nonparametric bootstrap are reasonably close to the empirical standard error and the procedure demonstrates acceptable coverage. The empirical standard deviation is generally higher for the cubic spline method compared to the linear method, which is not surprising given our allowed flexibility via splines and resulting number of parameters which must be estimated from the data. Of course, when we truly do not know the form of the true underlying function, such flexibility is important to consider. Simulation results for all settings where $K=25$ and $n=10$, displayed in the bottom portion of Table \ref{sim_results}, demonstrate similar good finite sample performance. Results using the PAB procedure are shown and discussed in Appendix B. 

\section{Data Application} \label{sec:application}
 We applied our proposed method to data from clinical trials in schizophrenia in which patients who were treated by 198 psychiatrists where randomly assigned to either the experimental treatment (risperidone) or an active control (haloperidol); this dataset, \texttt{Schizo}, is available in the \texttt{Surrogate} package on CRAN \citep{Surrogate}. We restricted our analysis to the 28 psychiatrists who treated at least 6 patients in each treatment group, and treat the psychiatrists as our ``studies". The primary outcome of interest was change in the Positive and Negative Syndrome Scale (PANSS) score, obtained from a standardized 30-item instrument used to rate the severity of symptoms in patients with schizophrenia \citep{kay1987positive}. The PANSS includes explicit positive, negative, and general psychopathology subscales and is considered the gold standard in schizophrenia clinical trials \citep{opler2017positive}.
Our surrogate marker of interest was change in the Brief Psychiatric Rating Scale (BPRS) score, obtained from an 18-item instrument \citep{overall1962brief}. The BPRS focuses on core psychopathology, including hallucinations, delusions, affect, and anxiety, and has a shorter administration time than the PANSS \citep{bell1992positive,kumari2017assessment}. This lower measurement burden motivates interest in the BPRS as a potential surrogate for PANSS \citep{alonso2002investigating,alonso2018maximum}. Both changes are calculated as the score at the start of treatment minus the score at the end of treatment. Higher values of the PANSS and BPRS scores indicate more severe symptoms; thus a positive change indicates a decrease in the score, and reflects a better outcome. That is, higher values of $Y$ and $S$ in this example are ``better".  

For the 28 studies (psychiatrists), both $Y$ and $S$ are measured. For our illustration, we pulled out two studies with ID 50 and ID 3  to illustrate as our $K+1$ study, resulting in $K=26$ prior studies. Study 50 was selected because the estimated treatment effect on BPRS was positive, while the estimated treatment effect on PANSS was negative i.e., the surrogate paradox was present. Study 3 was selected as a comparison; here the estimated treatment effect was positive for both BPRS and PANSS. Thus, we would at the very least hope that our proposed method would result in a \textit{higher} (more risky) probability of the surrogate paradox, the resilience probability, for Study 50 than for Study 3. 

Figure \ref{fig:data} shows the relationship between the BPRS score and the PANSS score in each of the 26 studies within the control group and the treated group.  We applied our proposed methods using cubic splines for the specified
mean function. The estimated resilience probability for Study 50 was 0.15 (SE$=0.345$), which was higher, as expected, than the estimated resilience probability, for Study 3 which was 0.02 (SE$=0.308$). Overall, our results demonstrate that the proposed method can appropriately differentiate studies with a positive treatment effect on both the surrogate and outcome from those with a positive treatment effect on the surrogate, but a negative effect on the outcome i.e., a surrogate paradox situation. In particular, the higher estimated resilience probability for Study 50 and the notably low probability for Study 3 align with the underlying treatment effect patterns and support the practical utility of the approach.

\vspace*{-.3in}
\section{Discussion}
\vspace*{-.1in}
In this work, we have proposed and empirically evaluated a functional class-based method to assess resilience to the surrogate paradox in a meta-analytic setting. Our method leverages data from $K$ completed studies in which the surrogate marker and primary outcome have been measured to make inference on a new study in which only the surrogate is measured. Unlike previous methods, we do not assume direct transportability of the conditional mean functions from the completed studies to the new study and instead, our approach considers deviations of functions from those observed in the completed studies to estimate the resilience probability. Our proposed methods are implemented in the R package \texttt{SurrogateParadoxTest} available at \url{https://github.com/emily13hsiao/SurrogateParadoxTest}.

Given that our proposed approach involves generating a synthetic distribution of treatment effects, it is worth discussing connections to longstanding research on synthetic data as well as more recent work on synthetic data generation using generative models. Synthetic data generation has long been an important area of statistical research, often motivated by the goal of producing privacy-preserving artificial datasets that approximate the joint distribution of the observed data \citep{raghunathan2021synthetic}. In parallel, developments in machine learning have produced increasingly flexible generative models capable of encoding complex dependence structures and using these representations to generate synthetic data \citep{harshvardhan2020comprehensive,zhou2025emerging}. These advances suggest interesting opportunities for incorporating data-adaptive representations of the relationships between treatment, surrogate, and outcome within surrogate research. For example, in our setting, generative modeling approaches could be used to learn flexible, joint representations of these variables that are compatible with the observed data from the $K$ prior studies, and then leverage these representations to generate synthetic data over which the surrogate paradox risk can be evaluated. More broadly, tools developed for synthetic data generation may be useful in future work to inform not only the assessment of surrogate paradox risk, but also possibly, the identification, validation, and principled use of surrogate markers in complex settings. 

In addition, an interesting direction for future work would be to extend our proposed framework to incorporate baseline covariates into the modeling of the conditional mean functions \citep{shafie2023incorporating}. Allowing these mean functions to depend on baseline covariates, potentially in study-specific ways, would make it possible to capture heterogeneity across studies in how the surrogate and outcome relationship may depend on patient characteristics. Such an extension could also be particularly valuable when baseline covariate information is available for the new study, as it would allow the method to assess how closely the new study resembles each of the prior studies in terms of its covariate distribution. This could further enable data-adaptive weighting of prior studies when constructing the resilience probability estimate, placing greater emphasis on those whose covariate–outcome relationships that are most compatible with the future study. While we do not pursue this extension here, incorporating covariate-dependent and study-specific structure into the resilience probability framework represents a promising avenue for better understanding transportability and robustness of surrogate-based conclusions across studies.


Our proposed method has some limitations. First, it is not guaranteed that the observed functions in each study are generated by a Gaussian Process as specified; in fact, this is a rather strong assumption. Furthermore, our method makes the assumption that the observed values of the surrogate are fixed in all prior studies, and that deviations in the potential $\Delta_{K+1}^*$ values are driven by randomness only in the generated functions, which may not necessarily be the case if some underlying mechanism causes the two quantities to vary jointly. Future work could involve including randomness in the values of $S$ to drive randomness in $\Delta_{K+1}$, as well as including covariates in the modeling of the GP mean function, as noted above, or weighting different studies based on distance of population covariates. 

A further limitation is that the proposed estimation procedure can exhibit numerical instability in finite samples. In some settings, near-singular covariance matrices or sensitivity to kernel hyperparameters can lead to instability in function estimation and, consequently, in the resulting resilience probability estimates. While we mitigate these issues through careful implementation, additional work on more stable estimation strategies or alternative function classes may be warranted. Finally, the computational burden of the method may be nontrivial when the number of studies or study sizes are large. The repeated fitting and simulation required to characterize the induced distribution of treatment effects can be expensive, potentially limiting scalability in high-dimensional or large-sample settings. Developing more efficient approximations or reduced-complexity representations that preserve the essential features of the functional uncertainty remains an important direction for future work.

In summary, this work introduces a flexible, functional-class framework for assessing resilience to the surrogate paradox that relaxes strong transportability assumptions commonly made in meta-analytic settings. By explicitly accounting for uncertainty in the relationship between the surrogate and the primary, our approach provides a principled and interpretable way to quantify surrogate paradox risk in future studies that intend to use the surrogate to test for a treatment effect. We view this framework as a useful contribution toward more robust and transparent use of surrogates for testing in future studies and a foundation for future methodological developments in this area.

\section*{Acknowledgments and Funding}
This work was supported by NIDDK grant R01DK118354 (PI:Parast).

\clearpage

\begin{figure}[htbp]
    \centering
           \includegraphics[width=0.8\linewidth]{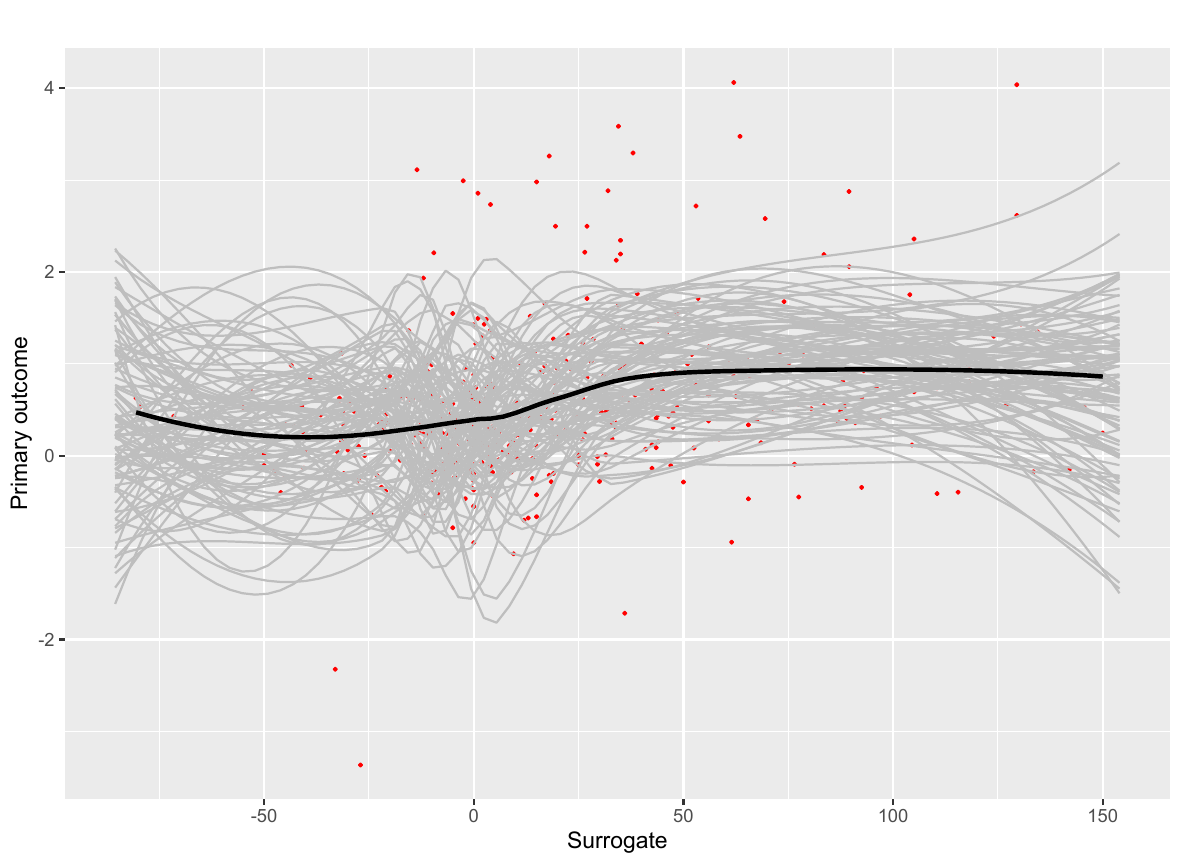}

     \caption{Functional class synthetic data generation; the red points indicate individual-level data, the black line reflects the estimated $m_g$, each grey line is a randomly generated conditional mean function }    \label{fig:idea}
\end{figure}

\clearpage \begin{figure}[htbp]
    \centering
           \includegraphics[width=0.8\linewidth]{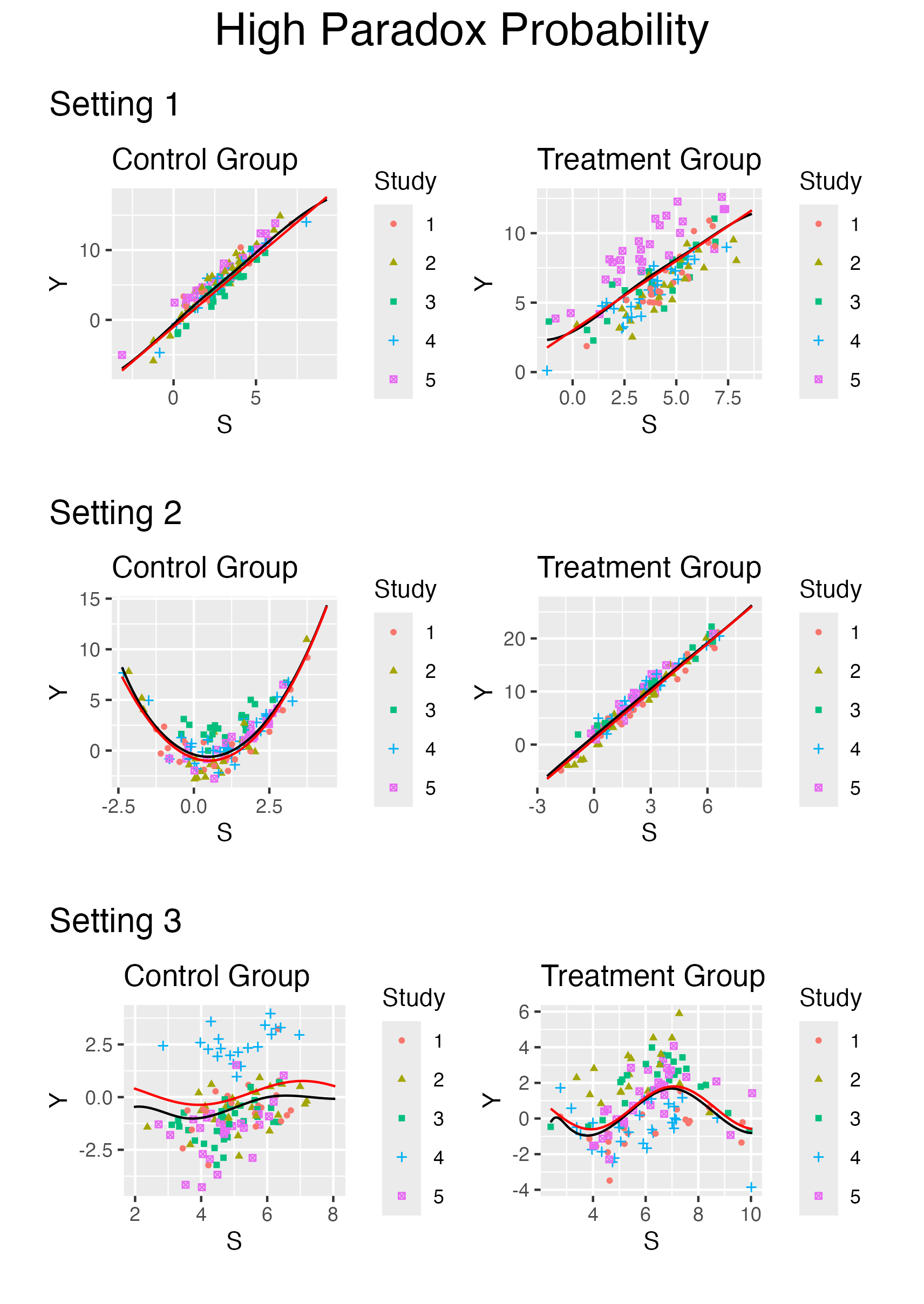}
          \label{fig:sub1}
     \caption{Illustration of simulation Settings 1-3 which were designed to have a high probability of the surrogate paradox in a new study; the subfigures show a scatterplot of the observed data from multiple studies with each study in a different color, the black line represents the true underlying mean function and the red line represents the cubic B-splines estimated mean function based on the studies observed;  in each setting, 10 studies of sample size 100 were generated, though the overlaid scatterplots show only 5 studies with 25 randomly sampled points for ease of visibility. }
    \label{fig:high}
\end{figure}

\begin{figure}[htbp]
    \centering
           \includegraphics[width=0.8\linewidth]{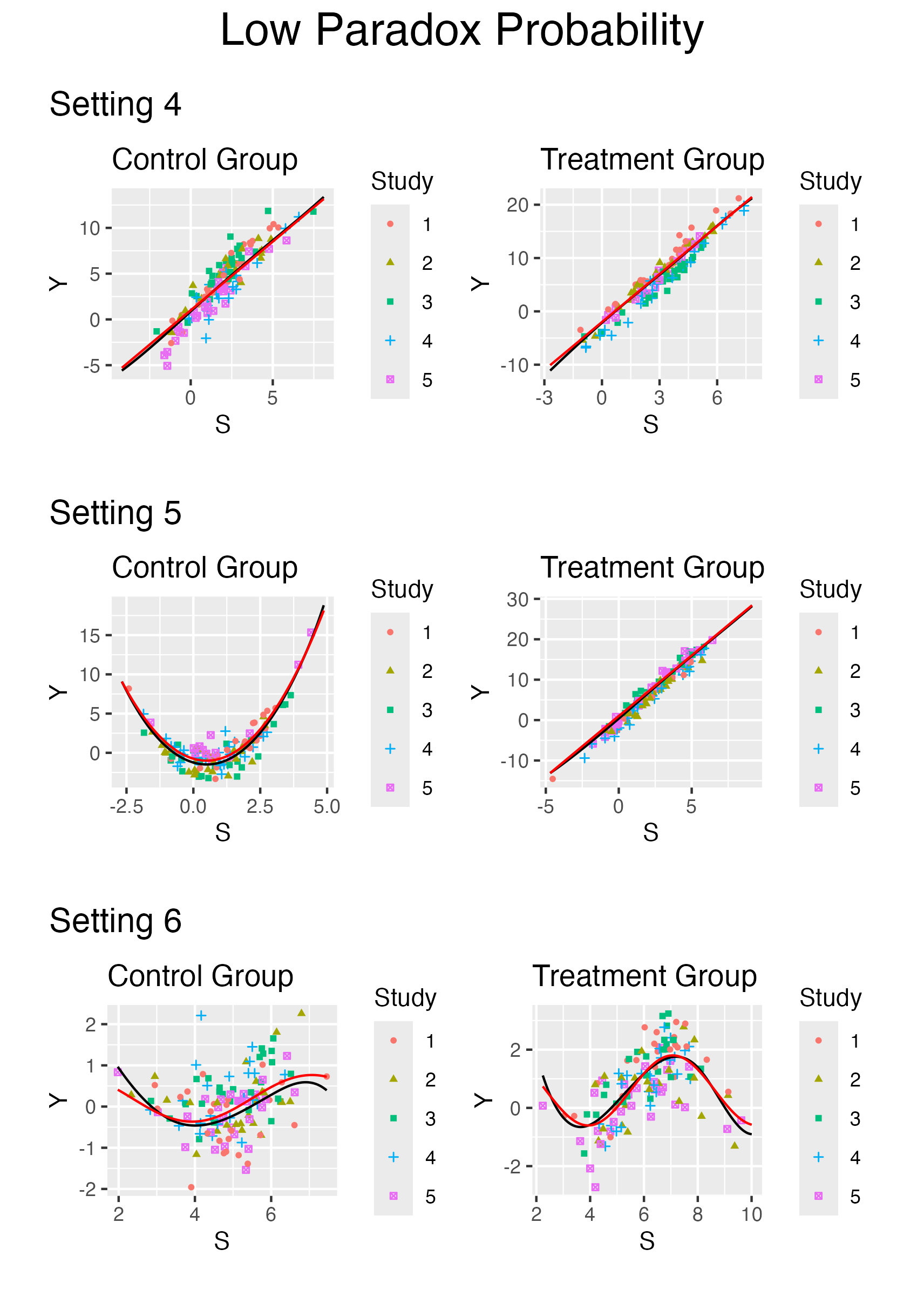}
          \label{fig:sub1}
     \caption{Illustration of simulation Settings 4-6 which were designed to have a low probability of the surrogate paradox in a new study; the subfigures show a scatterplot of the observed data from multiple studies with each study in a different color, the black line represents the true underlying mean function and the red line represents the cubic B-splines estimated mean function based on the studies observed;  in each setting, 10 studies of sample size 100 were generated, though the overlaid scatterplots show only 5 studies with 25 randomly sampled points for ease of visibility. }
    \label{fig:low}
\end{figure}

\clearpage

\begin{table}[H]
\centering
\small
\setlength{\tabcolsep}{3.5pt}

\begin{tabular}{l c| cccc | cccc | cccc | c}
\hline
\multicolumn{15}{c}{\textbf{$\mathbf{K=10}$, sample size per study and arm $\mathbf{=100}$}} \\
\hline
Setting & Truth
& \multicolumn{4}{c|}{Linear}
& \multicolumn{4}{c|}{Cubic}
& \multicolumn{4}{c}{Cubic Spline}
& \multicolumn{1}{|c}{MV} \\
\hline
 & & Est  & ESE & ASE & CP
   & Est  & ESE & ASE & CP
   & Est & ESE & ASE & CP
   & Est \\
\hline
1    &0.586&0.568&0.138&0.133&0.914&0.567&0.145&0.138&0.914&0.568&0.146&0.138&0.915&0.467\\ 
2    &0.77 &0.864&0.029&0.03 &0.167&0.757&0.156&0.144&0.925&0.757&0.156&0.143&0.924&0.367\\ 
3    &0.507&0.401&0.12 &0.108&0.762&0.426&0.14 &0.129&0.86 &0.481&0.139&0.128&0.885&0.486\\ 
4    &0.029&0.033&0.035&0.033&0.892&0.034&0.039&0.036&0.895&0.034&0.038&0.036&0.894&0.194\\ 
5    &0.078&0.838&0.033&0.036&0    &0.082&0.078&0.076&0.901&0.082&0.078&0.075&0.903&0.375\\ 
6    &0    &0.433&0.073&0.069&0    &0.069&0.07 &0.045&0.674&0.01 &0.013&0.013&0.980 &0.193\\ 
\hline
\multicolumn{15}{c}{\textbf{$\mathbf{K=25}$, sample size per study and arm $\mathbf{=10}$}} \\
\hline
Setting & Truth
& \multicolumn{4}{c|}{Linear}
& \multicolumn{4}{c|}{Cubic}
& \multicolumn{4}{c}{Cubic Spline}
& \multicolumn{1}{|c}{MV} \\
\hline
 &Truth & Est  & ESE & ASE & CP
   & Est  & ESE & ASE & CP
   & Est & ESE & ASE & CP
   & Est \\
\hline
1    &0.592&0.552&0.192&0.17 &0.906&0.554&0.193&0.172&0.922&0.552&0.192&0.173&0.922&0.306\\ 
2    &0.651&0.603&0.182&0.161&0.903&0.684&0.291&0.231&0.893&0.686&0.288&0.232&0.891&0.151\\ 
3    &0.502&0.392&0.094&0.09 &0.754&0.455&0.116&0.11 &0.94 &0.481&0.105&0.105&0.941&0.444\\ 
4    &0.076&0.071&0.105&0.1  &0.85 &0.071&0.106&0.1  &0.87 &0.071&0.106&0.101&0.858&0.143\\ 
5    &0.169&0.434&0.185&0.16 &0.697&0.169&0.214&0.182&0.884&0.169&0.213&0.183&0.878&0.143\\ 
6    &0.02 &0.203&0.079&0.072&0.083&0.056&0.082&0.093&0.968&0.023&0.03 &0.03 &0.948&0.147\\ 
\hline
\end{tabular}
\caption{Simulation results examining the bias and variance of estimated $\widehat{p}$ values for Settings 1-6, over 1000 iterations, with a sample size of 100 patients in each study in each group, and 10 studies (top portion) and with a sample size of 10 patients in each study in each group, and 25 studies (bottom portion). The first column represents the true probability of the paradox in the new study. The Est column is the average and the ESE column is the empirical standard deviation of the $\widehat{p}$ values obtained using our proposed method, using linear, cubic, or cubic splines, summarized over 1000 simulation iterations. The ASE column indicates the average standard error and CP indicates the coverage of the 95\% confidence intervals using the fully nonparametric bootstrap. Recall that in Settings 1 and 4, the conditional functions are truly linear; in Settings 2 and 5, the conditional functions are linear/quadratic; in Settings 3 and 6, the conditional functions are neither linear or cubic. The column MV compares to the method of \citet{elliott2015surrogacy}.} \label{sim_results}
\end{table}

\normalsize
\begin{figure}[htbp]
    \centering
           \includegraphics[width=0.8\linewidth]{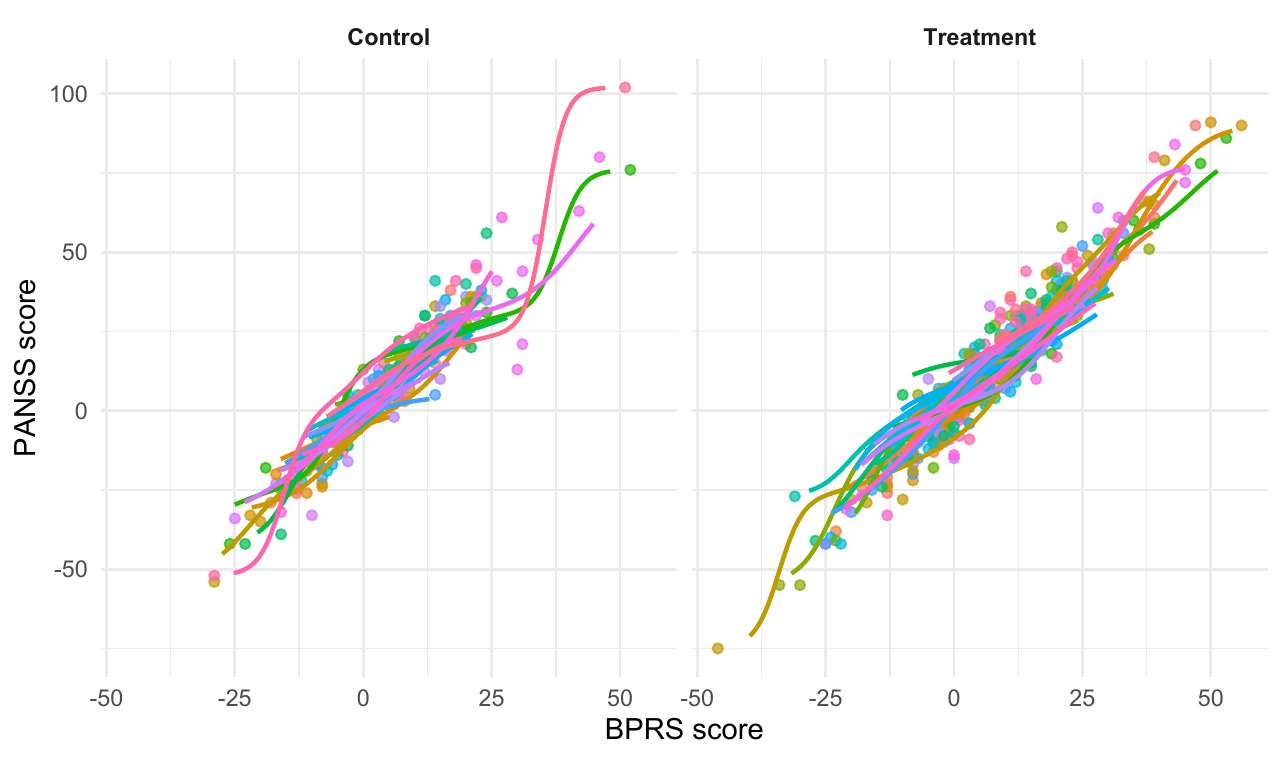}

     \caption{Relationship between the Brief Psychiatric Rating Scale (BPRS) score (surrogate) and Positive and Negative Syndrome Scale (PANSS) score (primary outcome) in each of the 26 studies (each color/curve is a different study) in the schizophrenia data application, in the control group (left side) and the treated group (right side) \\}    \label{fig:data}
\end{figure}

\def\bW{{\bf W}}
\def\bZ{{\bf Z}}
\def\bL{{\bf L}}
\def\bw{{\bf w}}
\def\supone{^{(1)}}
\def\supzero{^{(0)}}
\def\supg{^{(g)}}
\def\transpose{{\sf \scriptscriptstyle{T}}}
\def\bU{{\bf U}}
\def\bUhat{{\widehat \bU}}
\def\trans{^{\transpose}}
\def\var{\text{Var}}
\def\bbeta{\boldsymbol{\beta}}
\def\bomega{\boldsymbol{\omega}}
\def\bPi{\boldsymbol{\Pi}}
\def\calt{\mathcal{T}}
\setcitestyle{authoryear}
\def\bSigma{\boldsymbol{\Sigma}}
\def\bmu{\boldsymbol{\mu}}

\setlength\parindent{15pt}

\clearpage
\appendix

\setcounter{table}{0}
\renewcommand{\thetable}{A\arabic{table}}
\setcounter{figure}{0}
\renewcommand{\thefigure}{A\arabic{figure}}
\renewcommand{\theequation}{A.\arabic{equation}}

\allowdisplaybreaks

\section*{Appendix A}
 \setlength{\parindent}{0cm}
\textit{Proof of Theorem 1.}   As noted in the main text, the general steps of this proof will be to (1) express $p_0$ as a functional of unknown parameters $\Gamma_g$, basis components, and the true cumulative distribution functions (CDFs) of $S_{K+1}\supzero$ and   $S_{K+1}\supone$; (2) prove consistency and asymptotic normality of the estimates of the parameter $\Gamma_g$ using data from the $K$ studies via our proposed procedure; and (3) use these asymptotic results, combined with the smoothness of the functional defining $p_0$ to establish consistency and asymptotic normality of $\widehat{p}$. 

\textit{Step 1: Express $p_0$ as functional.} Considering the true conditional mean in each group $g=0,1$ in Study K+1, by Assumption \ref{asmp:FC}, 
$$E(Y\supg_{K+1} | S\supg_{K+1}=s) = \mu_{(K+1),g}(s) =  \beta_g' B(s) =  \sum_{l=1}^L \beta_{gl} B_l(s),$$
where $B(s)=(B_1(s),\ldots,B_L(s))'$ is a known basis (e.g., cubic splines in our implementation) and $L$ is the number of basis
components, and coefficients $\beta_g = (\beta_{g1},\ldots,\beta_{gL})'$. 
Thus, the resilience probability $p_0$ can be expressed as: 
\begin{align*}
 p_0=&P\left(\Delta_{K+1}<0\right)\\
 =&P\left(\int \mu_{K+1,1}(s)dF_{S_{K+1}^{(1)}}(s)-\int \mu_{K+1,0}(s)dF_{S_{K+1}^{(0)}}(s) <0\right)\\
 =&P\left(\int \sum_{l=1}^L \beta_{1l} B_l(s) dF_{S_{K+1}^{(1)}}(s)-\int \sum_{l=1}^L \beta_{0l} B_l(s) dF_{S_{K+1}^{(0)}}(s) <0\right)\\
 =&P\left( \sum_{l=1}^L \beta_{1l} \int B_l(s)dF_{S_{K+1}^{(1)}}(s)- \sum_{l=1}^L \beta_{0l}\int B_l(s)dF_{S_{K+1}^{(0)}}(s)<0\right) 
 \end{align*}
 Importantly, although $p_0$ can be expressed purely in terms of the coefficients $\beta_g$ (and basis components) as above, 
we estimate the coefficients jointly with the Gaussian Process covariance parameters $\phi_g=(\sigma_g^2,\theta_g)$ by maximum likelihood. Consequently, the sampling distribution of $\widehat p$ (and hence its asymptotic variance) will depend on both $\beta_g$ and $\phi_g$. Thus, we will make explicit our reliance on these parameters by defining $\Gamma_g=(\beta_g,\phi_g) = (\beta_g,\sigma_g^2,\theta_g)$ and expressing $p_0$ as
\begin{align*}
 p_0  =& h\left(\Gamma_0, \Gamma_1, F_{S_{K+1}^{(0)}},\,F_{S_{K+1}^{(1)}} \right)
 \end{align*}
where $h(\cdot)$ is a differentiable function. That is, $p_0$ is a functional that depends jointly on the basis coefficients \(\beta_g\), the covariance parameters \(\phi_g\) (through the GP specification), and the surrogate CDFs.

\textit{Step 2: Consistent estimation of $\Gamma_g$.} Recall from the main text that within each group we obtain $(\widehat\beta_g,\widehat\phi_g)$ by maximizing the log-likelihood pooled across the $K$ training studies.  Dropping the group subscript $g$ for clarity, let the per-study log-likelihood be
$$
\ell_k(\beta,\phi)
= -\tfrac{1}{2}\log\big|C_k(\phi)\big|
  -\tfrac{1}{2}(Y_k - B(S_k)\beta)' C_k(\phi)^{-1}(Y_k - B(S_k)\beta),
$$
and define the empirical objective
$$
M_K(\beta,\phi) \;=\; \frac{1}{K}\sum_{k=1}^K \ell_k(\beta,\phi).
$$
The MLE is then the M-estimator
$$
(\widehat\beta_K,\widehat\phi_K)
\;=\; \arg\max_{(\beta,\phi)\in\Theta} M_K(\beta,\phi),
$$
where $\Theta$ is the compact parameter space for $(\beta,\phi)$ (see Assumption \ref{asmp:regularity}). Let the population objective be the expectation of the per-study log-likelihood,
\[
M(\beta,\phi) \;=\; E\big[\ell_k(\beta,\phi)\big],
\]
where the expectation is taken with respect to the study-level data-generating mechanism, and suppose $M(\beta,\phi)$ has a unique maximizer $(\beta^*,\phi^*)$ on $\Theta$. We aim to show that under Assumptions \ref{asmp:ind}--\ref{asmp:covregularity} which include the MLE regularity conditions in Assumption \ref{asmp:regularity}, 
\begin{equation}
(\widehat\beta_K,\widehat\phi_K) \xrightarrow{p} (\beta^*,\phi^*)\quad\text{as }K\to\infty, \label{consistency}
\end{equation}
where $(\beta^*,\phi^*)=\arg\max_{(\beta,\phi)\in\Theta} M(\beta,\phi)$.  Since $(\widehat\beta_K,\widehat\phi_K)$ is an M-estimator, we can leverage established M-estimation theory \citep{van2000asymptotic} and claim consistency, if we can prove the following uniform convergence:

\begin{equation}
 \sup_{\beta, \phi} \Bigg| \sum_{k=1}^K \ell_k(\beta, \phi) - \sum_{k=1}^K E[\ell_k(\beta, \phi)] \Bigg| \overset{p}{\to} 0, \label{uniform}
\end{equation}
i.e., the observed log likelihood converges uniformly to the population log likelihood. To this end, we will use the Uniform Law of Large Numbers for non-iid data \citep{van1996weak}[Thm.\ 2.8.1] which states that (\ref{uniform}) holds if (a) the data from the $K$ studies are independent, (b) Uniform Lipschitz Continuity holds i.e., for all $(\beta,\phi),(\beta',\phi')\in\Theta$ and all $k$,
\begin{equation}
\big|\ell_k(\beta,\phi)-\ell_k(\beta',\phi')\big|
\;\le\;
L_k \,\|(\beta,\phi)-(\beta',\phi')\|, \label{lip_condition}
\end{equation}
for a Lipschitz constant $L_k$, and (c) the envelope condition holds, i.e., there exist measurable functions $G_k$ such that
\[
\ell_k(\beta,\phi)^2 \;\le\; G_k \qquad \text{for all }(\beta,\phi)\in\Theta,
\]
and
\[
\sup_{K}\frac{1}{K}\sum_{k=1}^K \mathbb{E}\big[ G_k \big] < \infty.
\]

Conditions (a) and (c) are satisfied by our stated assumptions in the main text; it remains to verify condition (b) of Lipschitz continuity. We decompose the first term of (\ref{lip_condition}) as: 
  \begin{eqnarray}
    \ell_k(\beta, \phi) - \ell_k(\beta', \phi')
    &=& -\frac12 \log|C_k(\phi)|
        -\tfrac{1}{2}(Y_k - B(S_k)\beta)' C_k(\phi)^{-1}(Y_k - B(S_k)\beta) \nonumber \\
    &&\quad + \frac12 \log|C_k(\phi')|
        +\tfrac{1}{2}(Y_k - B(S_k)\beta')' C_k(\phi')^{-1}(Y_k - B(S_k)\beta')\nonumber  \\
    &=& -\frac12\big( \log|C_k(\phi)| - \log|C_k(\phi')| \big) \label{logdet} \\
    &&\quad - \frac12\Big[
         (Y_k - B(S_k)\beta)' C_k(\phi)^{-1} (Y_k - B(S_k)\beta) \nonumber \\
         && \quad- (Y_k - B(S_k)\beta')' C_k(\phi')^{-1} (Y_k - B(S_k)\beta')
       \Big], \label{quad}
\end{eqnarray}
and first derive a Lipschitz bound for (\ref{logdet}) term. Because we intend to utilize the multivariate mean value theorem, we apply the Jacobi formula for the derivative of the determinant:
\[
\frac{d}{dt} \det A(t)
= \det(A(t)) \, \mathrm{tr}\!\left(A(t)^{-1} \frac{dA(t)}{dt}\right).
\]

to  \(C_k(\phi)\), to obtain
\begin{eqnarray*}
\frac{d}{d\phi_j} \log |C_k(\phi)|
&=& \frac{1}{|C_k(\phi)|} 
  \frac{d}{d\phi_j} |C_k(\phi)|\\
  &=& \frac{1}{|C_k(\phi)|} |C_k(\phi)| \mathrm{tr} 
  \left(C_k(\phi)^{-1} \frac{d C_k(\phi)}{d\phi_j}\right)\\
&=& \mathrm{tr}\!\left(C_k(\phi)^{-1} \frac{d C_k(\phi)}{d\phi_j}\right).
\end{eqnarray*}

Using this expression, the multivariate mean value theorem, and properties of matrix traces, it follows that 
\begin{align*}
    \log|C_k(\phi)| - \log|C_k(\phi')| 
    &= \sum_j  \mathrm{tr} \left( C_k(\tilde{\phi})^{-1} \frac{\partial C_k(\tilde{\phi})}{\partial \phi_j} (\phi_j - \phi_j') \right) \\
    &= \mathrm{tr} \left( \sum_j C_k(\tilde{\phi})^{-1} \frac{\partial C_k(\tilde{\phi})}{\partial \phi_j} (\phi_j - \phi_j') \right) \\
    &= \mathrm{tr} \left( C_k(\tilde{\phi})^{-1} \sum_j \frac{\partial C_k(\tilde{\phi})}{\partial \phi_j} (\phi_j - \phi_j')\right)
\end{align*}
for some $\tilde{\phi}$ on the line segment between $\phi$  and $\phi'$. Next, we use the von Neumann trace-norm inequality together with the Cauchy-Schwarz inequality to state that:

\[
|\mathrm{tr}(AB)|
\le \sum_{i=1}^n \alpha_i \beta_i.
\le 
\left( \sum_{i=1}^n \alpha_i^2 \right)^{1/2}
\left( \sum_{i=1}^n \beta_i^2 \right)^{1/2}
= \|A\|_F \, \|B\|_F.
\]

where $\|A\|_F$ denotes the Frobenius norm of $A$ and thus,

\[
|\mathrm{tr}(AB)| \le \|A\|_F \, \|B\|_F.
\]

Applying this inequality, it follows that:
\begin{eqnarray*}
\Big |\log|C_k(\phi)| - \log|C_k(\phi')| \Big | &=& 
\Big | \mathrm{tr} \left( C_k(\tilde{\phi})^{-1} \sum_j \frac{\partial C_k(\tilde{\phi})}{\partial \phi_j} (\phi_j - \phi_j')\right)\Big | \\
&=& 
\leq \norm{C_k(\tilde{\phi})^{-1}}_F  \norm{\sum_j \frac{\partial C_k(\tilde{\phi})}{\partial \phi_j} (\phi_j - \phi_j')}_F
\end{eqnarray*}
Now we return to our eigenvalue bound assumption within Assumption 5, which states that all eigenvalues of \(C_k(\phi)\) are bounded above and below by $\lambda_{min}>0$ and $\lambda_{max}<\infty$. Therefore, $$\norm{C_k(\tilde{\phi})^{-1}}_F \leq \frac{\sqrt{n_k}}{\lambda_{min}}.$$

Next, we address the term: $$\norm{\sum_j(\phi_j - \phi_j') \frac{\partial C_k(\tilde{\phi})}{\partial \phi_j}}_F.$$ Define $M_k = \sup_{\theta} \norm{\frac{\partial C_k(\phi)}{\partial \theta_j}} < \infty$ which exists by assuming that $C_k(\phi)$ is continuously differentiable on a compact parameter space (within Assumption 5). By properties of the Frobenius norm, we have that 
\begin{align*}
    \norm{\sum_j(\phi_j - \phi_j') \frac{\partial C_k(\tilde{\phi})}{\partial \phi_j}}_F 
    &\leq \sum_j |\phi_j - \phi_j'| \norm{\frac{\partial C_k(\tilde{\phi})}{\partial \phi_j}}_F \\
    &\leq M_k  \sqrt{l}  \norm{\phi-\phi'}_2
\end{align*}
where $l$ is the length of $\phi$. Putting these two together, it follows that
$$
\Big |\log | C_k(\phi)| - \log | C_k(\phi')| \Big | \leq \frac{\sqrt{n_k}}{\lambda_{min}}  M_k  \sqrt{l} \norm{\phi - \phi'}_2
$$
giving a Lipschitz bound for the (\ref{logdet}) term with $L_k = \frac{M_k \sqrt{n_kl}}{\lambda_{min}}$.

Next, we derive a Lipschitz bound for the (\ref{quad}) term:
$$ 
         (Y_k - B(S_k)\beta)' C_k(\phi)^{-1} (Y_k - B(S_k)\beta) - (Y_k - B(S_k)\beta')' C_k(\phi')^{-1} (Y_k - B(S_k)\beta'),
       $$
dropping the constant $- \frac12$ for simplicity. We define:
\[
r_k(\beta) = Y_k - B(S_k) \beta, \qquad 
A_k(\phi) = C_k(\phi)^{-1},
\]
\[
Q_k(\beta,\phi)
    = r_k(\beta)^{\top} A_k(\phi)\, r_k(\beta),
\]
and we aim to bound:
\[
\Delta Q_k 
    \;=\; Q_k(\beta,\phi) - Q_k(\beta',\phi').
\]
which we can re-write as
\begin{eqnarray*}
\Delta Q_k
    &=& r^\top A r - r'^{\top} A' r'\\
    &=& r^\top (A-A') r
    \;+\; (r-r')^\top A' r
    \;+\; (r')^\top A' (r-r'),
\end{eqnarray*}
where \(r=r_k(\beta)\), \(r'=r_k(\beta')\), \(A=A_k(\theta)\), and \(A'=A_k(\theta')\), and the second equality follows from some algebra. It follows that 
\begin{eqnarray*}
|\Delta Q_k|
&\le& 
\bigl|r^\top (A-A') r\bigr|
+ \bigl|(r-r')^\top A' r\bigr|
+ \bigl|(r')^\top A' (r-r')\bigr|\\
&\le&
\|A - A'\|_2 \, \|r\|_2^2
\;+\;
\|A'\|_2 \, \|r-r'\|_2 \, \|r\|_2
\;+\;
\|A'\|_2 \, \|r'\|_2 \, \|r-r'\|_2\\
&=&
\|A - A'\|_2 \, \|r\|_2^2
+
\,\|A'\|_2 \,\|r-r'\|_2 (\|r\|_2 + \|r'\|_2) \
\end{eqnarray*}

where the second inequality follows from the matrix norm bound  
\[
|x^\top A y| \le \|A\|_2 \,\|x\|_2 \,\|y\|_2,
\]
and the third equality involves combining terms. Examining $\norm{A-A'}_2$, since
$$ A-A'
    = C_k(\phi)^{-1} - C_k(\phi')^{-1} = C_k(\phi)^{-1}[C_k(\phi') - C_k(\phi)] C_k(\phi')^{-1}
$$
it follows that,
\begin{align*}
  \norm{A-A'}_2   
    &\leq \norm{C_k(\phi)^{-1}}_2  \norm{C_k(\phi')-C_k(\phi)}_2  \norm{C_k(\phi')^{-1}} \\
    &\leq \left(\frac{1}{\lambda_{min}}\right)^2 L_{\phi} \norm{\phi-\phi'}_2 = \frac{L_{\phi}}{\lambda_{min}^2} \norm{\phi-\phi'}_2
\end{align*}
by the assumptions of boundedness of the eigenvalues of $C_k$  and $C_k(\phi)$ being Lipschitz in $\phi$, within Assumption 5, and where \(L_\phi\) is the corresponding Lipschitz constant of \(C_k(\phi)\) in \(\phi\).  Examining $\norm{r - r'}_2$, it follows that 
\begin{align*}
    \norm{r - r'}_2 
    &= \norm{Y_k - B(S_k)\beta - (Y_k - B(S_k)\beta')}_2 \\
    &= \norm{B(S_k)(\beta - \beta')}_2 \\
    &\leq \norm{B(S_k)}_2  \norm{\beta - \beta'}_2 \\
    &\leq B^* \norm{\beta - \beta'}_2\\
\end{align*}
where the last line follows from Assumption 4 which includes $\norm{B(S_k)}_2 \leq B^*$. Examining $\norm{r}_2 + \norm{r'}_2$,
\begin{align*}
    \norm{r}_2 + \norm{r'}_2
    &= \norm{Y_k - B(S_k)\beta} + \norm{Y_k - B(S_k) \beta'} \\
    &\leq \norm{Y_k}_2 + \norm{B(S_k)}_2  \norm{\beta}_2 + \norm{Y_k} + \norm{B(S_k)}_2 \norm{\beta'}_2 \\
    &= 2 \norm{Y_k}_2 + B^*  (\norm{\beta}_2 + \norm{\beta'}_2) 
\end{align*}
under Assumptions 4 and 5. In addition, since the eigenvalues of \(C_k(\phi)\) are bounded below by \(\lambda_{\min}>0\),
\[
\|A'\|_2
=
\|C_k(\phi')^{-1}\|_2
\le
\frac{1}{\lambda_{\min}}.
\]
Examining $\norm{r}_2^2$, using similar arguments as above, it follows that
\begin{align*}
    \norm{r}_2^2 
    &= \norm{Y_k - B(S_k)\beta}_2^2 \leq \norm{Y_k}_2^2 + 2B^*\norm{\beta}_2\norm{Y_k}_2 + (B^*\norm{\beta}_2)^2 
\end{align*}

 Finally, putting together all bounds, it follows that: 
\begin{align*}
    |\Delta Q_k| 
    &\leq \norm{A-A'}_2 \norm{r}_2^2 + \norm{A'}_2 \norm{r-r'}_2 (\norm{r}_2 + \norm{r'}_2) \\
    &\leq \left[ \frac{L_{\phi}}{\lambda_{\min}^2} \left( \norm{Y_k}_2^2 + 2B^*\norm{Y_k}_2\norm{\beta}_2 + (B^*\norm{\beta}_2)^2 \right) \right] \norm{\phi-\phi'}_2 \\
    &\quad + \left[ \frac{B^*}{\lambda_{\min}} \left( 2\norm{Y_k}_2 + B^*(\norm{\beta}_2 + \norm{\beta'}_2) \right) \right] \norm{\beta-\beta'}_2.
\end{align*} 
Since the parameters $(\beta, \phi)$ are assumed to lie in a compact space (Assumption 4), there exists a constant $R < \infty$ such that $\norm{\beta}_2 \leq R$ and $\norm{\phi}_2 \leq R$. Substituting this into the coefficients, we obtain
\begin{align*}
    |\Delta Q_k| &\leq L_1(Y_k) \norm{\phi-\phi'}_2 + L_2(Y_k) \norm{\beta-\beta'}_2 \\
    &\leq \mathcal{L}(Y_k) \norm{(\beta, \phi) - (\beta', \phi')}_2,
\end{align*}
where the stochastic Lipschitz coefficient $\mathcal{L}(Y_k)$ is a polynomial in $\norm{Y_k}_2$ of degree at most 2:
\[
\mathcal{L}(Y_k) = \sqrt{2} \max \left\{ \frac{L_{\phi}}{\lambda_{\min}^2} (\norm{Y_k}_2^2 + 2B^*R\norm{Y_k}_2 + (B^*R)^2), \frac{B^*}{\lambda_{\min}} (2\norm{Y_k}_2 + 2B^*R) \right\}.
\]
Under the assumption that $Y$ has a finite second moment (Assumption 4), $E[\norm{Y_k}_2^2] < \infty$, it follows that $E[\mathcal{L}(Y_k)] < \infty$. Thus, $Q_k(\beta, \phi)$ is Lipschitz continuous in the parameters with an integrable Lipschitz constant. Thus, we have shown that this quadratic term is also Lipschitz in the parameters $\beta$ and $\phi$, establishing that (\ref{lip_condition}) holds. It then follows that the needed uniform convergence in (\ref{uniform}) holds, and thus the consistency of our estimates stated in (\ref{consistency}) holds. Furthermore, under the MLE regularity conditions in Assumption \ref{asmp:regularity} and the results above, we have
\begin{equation}
\sqrt{K}\,\big(\widehat\Gamma_K - \Gamma^*\big)\ \xrightarrow{d}\ N\big(0,\Sigma_{\Gamma}\big),
\label{gamma_normal}
\end{equation}
where $\widehat\Gamma_K=(\widehat\beta_K,\widehat\phi_K)'$, $\Gamma^*=(\beta^*,\phi^*)'$, and $\Sigma_{\Gamma}$ is the asymptotic covariance matrix of the MLE, $\Sigma_{\Gamma} \;=\; A^{-1} B A^{-1},$
where, with a slight abuse of notation, we re-define $A$ and $B$ as:
$$
A \;=\; -\mathbb{E}\!\left[\nabla^2_{\Gamma}\,\ell_k(\Gamma^*)\right], 
\qquad
B \;=\; \mathbb{E}\!\left[\nabla_{\Gamma}\ell_k(\Gamma^*)\,\nabla_{\Gamma}\ell_k(\Gamma^*)^\top\right].
$$

\textit{Step 3: Properties of $\widehat{p}$.} Now that we have stated properties of the parameter estimates for each group, with a slight abuse of notation, we will re-introduce the subscript $g$ to denote the group $g=0,1$. From the above steps, we have that \begin{align*}
 p_0  =& h\left(\Gamma_0^*, \Gamma_1^*, F_{S_{K+1}^{(0)}},\,F_{S_{K+1}^{(1)}} \right)
 \end{align*}
where $\Gamma_g^* = (\beta_g^*, \phi_g^*)$ and $h(\cdot)$ is a differentiable function, and by construction 
\begin{align*}
 \widehat{p}  =& h\left(\widehat{\Gamma}_0, \widehat \Gamma_1, \widehat F_{S_{K+1}^{(0)}},\,\widehat F_{S_{K+1}^{(1)}} \right)
 \end{align*}
 where $\widehat F_{S_{K+1}^{(g)}}$ is the empirical CDF. Since $\widehat{\Gamma}_0 \rightarrow \Gamma_0^*$, $\widehat \Gamma_1 \rightarrow \Gamma_1^*$, and $\widehat F_{S_{K+1}^{(g)}}$ is a consistent estimator of the true CDF for $g=0,1$, it follows that 

$$  \widehat{p}  = h\left(\widehat{\Gamma}_0, \widehat \Gamma_1, \widehat F_{S_{K+1}^{(0)}},\,\widehat F_{S_{K+1}^{(1)}} \right) \rightarrow  h\left(\Gamma_0, \Gamma_1, F_{S_{K+1}^{(0)}},\,F_{S_{K+1}^{(1)}} \right) = p_0.$$
  
Next, we address asymptotic normality. Let $n_s = [n_{(K+1),1}n_{(K+1),0}]/n_{K+1}$ and let $\mathbf{u}(\widehat{F})$ and $\mathbf{u}(F)$ denote the vector of integrals from Step 1: $$\mathbf{u}(\widehat{F}) = \left (\int B_1(s)d\widehat{F}_{S_{K+1}^{(1)}}(s), ..., \int B_L(s)d\widehat{F}_{S_{K+1}^{(0)}}(s)\right )',$$ $$\mathbf{u}(F) =\left ( \int B_1(s)dF_{S_{K+1}^{(1)}}(s),..., \int B_L(s)dF_{S_{K+1}^{(0)}}(s)\right )'.$$ Under our stated assumptions and standard empirical process theory, 
$$\sqrt{n_s}\left ( \mathbf{u}(\widehat{F}) - \mathbf{u}(F)\right ) \xrightarrow{d} N(0,\Upsilon),$$
 for some asymptotic covariance matrix $\Upsilon$ that depends on the variability in the empirical CDFs and the basis functions $\{B_l\}$.  Next, we re-parameterize $h$ as $$H(\Gamma, \mathbf{u}(F)) = h\left(\Gamma_0^*, \Gamma_1^*, F_{S_{K+1}^{(0)}},\,F_{S_{K+1}^{(1)}} \right),$$ use the following to denote needed gradients:
 $$D_{\Gamma} = \frac{\partial H(\Gamma, \mathbf{u}(F)) }{\partial \Gamma}\bigg|_{(\Gamma^*, \mathbf{u}(F))}  \quad \mbox{and} \quad D_{\mathbf{u}}= \frac{\partial H(\Gamma, \mathbf{u}(F)) }{\partial \mathbf{u}}\bigg|_{(\Gamma^*, \mathbf{u}(F))}$$ and leverage the fact that the data used to estimate $\widehat{\Gamma}$ (the K prior studies) and the data used to estimate $\widehat{F}$ (the surrogate marker data from Study $K+1$) are independent, to apply a first-order Taylor expansion of $H$ about the true value $(\Gamma^*,F)$, 
 and use our stated properties to claim that:
\begin{eqnarray*}
\sqrt{n_s}(\widehat p - p_0) &=& \sqrt{n_s} D_{\Gamma} \left ( \widehat \Gamma -\Gamma^*\right) + \sqrt{n_s} D_{\mathbf{u}} \left ( \mathbf{u}(F) -\mathbf{u}(\widehat F)\right)\\
&=& \frac{\sqrt{n_s}}{\sqrt{K}} D_{\Gamma} \sqrt{K}\left ( \widehat \Gamma -\Gamma^*\right) + D_{\mathbf{u}} \sqrt{n_s} \left ( \mathbf{u}(F) -\mathbf{u}(\widehat F)\right)\\
&\xrightarrow{d}& N(0, \kappa D_{\Gamma}\Sigma_{\Gamma} D_{\Gamma}' + D_{\mathbf{u}} \Upsilon D_{\mathbf{u}}')
\end{eqnarray*}

where $\kappa = n_s/K$ such that $$\sigma_P^2 = \kappa D_{\Gamma}\Sigma_{\Gamma} D_{\Gamma}' + D_{\mathbf{u}} \Upsilon D_{\mathbf{u}}'.$$

\section*{Appendix B}
\setlength\parindent{15pt} 
\subsection*{Simulation Data Generation Details}
In Setting 1, the true mean functions were $m_0(s) = 2s-1$ and $m_1(s) = s + 3$, the Gaussian Process parameters were $\theta = 5, \sigma^2 = 1$, and $v^2 = 1$, and in the prior studies, $S\supzero \sim N(3, 3)$ and $S\supone \sim N(4, 3)$. In contrast, in Study $K+1$, data were generated such that $S\supzero \sim N(4.75, 1)$ and $S\supone \sim N(5.25, 1)$.  In Setting 2, the true mean functions used were $m_0(s) = (s - 0.5)^2 - 1$ and $m_1(s) = 3s + 1$, the Gaussian Process parameters were $\theta = 5$, $\sigma^2 = 1$, and $v^2 = 1$, and in the completed studies, $S\supzero \sim N(0.9, 1.5)$ and $S\supone \sim N(2.2, 4.5)$. In Study $K+1$, data were generated such that $S\supzero \sim N(-0.7, 1)$ and $S\supone \sim N(-0.2, 2)$.  In Setting 3, the true mean functions used were $m_0(s) = 0.2 + 0.4\sin(s) + 0.4\cos(s)$ and $m_1(s) = 0.6 + 0.85\sin(s) + 0.85\cos(s)$, the Gaussian Process parameters were $\theta = 5$, $\sigma^2 = 1$, and $v^2 = 1$, and in the completed studies  $S\supzero \sim N(5,1)$ and $S\supone \sim N(6,2)$. In Study $K+1$, data were generated such that  $S\supzero \sim N(4.1, 0.5)$ and $S\supone \sim N(4.1, 0.5)$. In this setting, the surrogate paradox was

 Setting 4, similar to Setting 1, reflects a simple linear model in which the distribution of the surrogate changes only slightly in the later study. In Setting 4, the true mean functions used were $m_0(s) = 1.5s + 1$ and $m_1(s) = 3s-2$, the Gaussian Process parameters were $\theta = 5$, $\sigma^2 = 1$, and $v^2 = 1$, and in the completed studies, $S\supzero \sim N(2,3)$ and $S\supone \sim N(3,3)$. In Study $K+1$, data were generated as $S\supzero \sim N(1.75, 1)$ and $S\supone \sim N(2.75, 1)$. In Setting 5, the true mean functions used were $m_0(s) = (s - 0.5)^2 - 1$ and $m_1(s) = 3s + 1$, the Gaussian Process parameters were $\theta = 5$, $\sigma^2 = 1$, and $v^2 = 1$, and in the completed studies, $S\supzero \sim N(0.9, 1.5)$ and $S\supone \sim N(2.2, 4.5)$. In Study $K+1$, data were generated as $S\supzero \sim N(-0.08, 1)$ and $S\supone \sim N(0.45, 2)$. In Setting 6, the true mean functions used were $m_0(s) = 0.2 + 0.4\sin(s) + 0.5\cos(s)$ and $m_1(s) = 0.6 + 0.85\sin(s) + 0.85\cos(s)$, the Gaussian Process parameters were $\theta = 5$, $\sigma^2 = 0.1$, and $v^2 = 0.5$, and in the completed studies, $S\supzero \sim N(5, 1)$ and $S\supone \sim N(6, 2)$.  In Study $K+1$, data were generated as $S\supzero \sim N(5.5, 0.5)$ and $S\supone \sim N(6.5, 0.5)$.  

\subsection*{Comparison to Elliott et al. (2015)}
In this section, we detail the MLE approach of \citet{elliott2015surrogacy} which is implemented in our simulation study for comparison. \citet{elliott2015surrogacy} specifies the following models:
\begin{eqnarray*}
S_{i0k} &=& \alpha_S + a_{S_k} + \epsilon_{S_{ki}} \\
S_{i1k} &=& \alpha_S + \beta_S + a_{S_k} + b_{S_k} + \epsilon_{S_{ki}} \\
Y_{i0k} &=& \alpha_Y + a_{Y_k} + \epsilon_{Y_{ki}} \\
Y_{i1k} &=& \alpha_Y + \beta_Y + a_{Y_k} + b_{Y_k} + \epsilon_{Y_{ki}} \\
\end{eqnarray*}
where
\begin{equation*}
    \begin{pmatrix}
\epsilon_{S_{ki}} \\
\epsilon_{Y_{ki}}
\end{pmatrix}
\sim
\mbox{BVN}\!\left(
\begin{pmatrix}
0 \\
0
\end{pmatrix},
\begin{pmatrix}
\sigma_{ss} & \sigma_{sy} \\
\sigma_{sy} & \sigma_{yy}
\end{pmatrix}
\right).
\end{equation*}
and the random trial-level effects are:
\begin{equation*}
    \begin{pmatrix}
a_{S_k} \\
a_{Y_k} \\ 
b_{S_k} \\
b_{Y_k} \\ 
\end{pmatrix}
\sim
\mbox{MVN}\!\left(
\begin{pmatrix}
0 \\
0\\
0\\
0
\end{pmatrix},
\begin{pmatrix}
d_{ss} & d_{sy} & d_{sa} & d_{sb}\\
 & d_{yy} & d_{ya} & d_{yb}\\
 &  & d_{aa} & d_{ab}\\
 &  &  & d_{bb}
\end{pmatrix}
\right).
\end{equation*}
and it follows that 
\begin{equation*}
    \begin{pmatrix}
\Delta_{Sk} \\
\Delta_{Yk}
\end{pmatrix}
\sim
\mbox{BVN}\!\left(
\begin{pmatrix}
\beta_S\\
\beta_Y
\end{pmatrix},
\begin{pmatrix}
d_{aa} & d_{ab} \\
d_{ab} & d_{bb}
\end{pmatrix}
\right).
\end{equation*}
where $\Delta_{Sk} = E(S_k\supone - S_k\supzero)$ and $\Delta_{Yk} = E(Y_k\supone - Y_k\supzero)$. The parameters $\beta_S,\beta_Y, d_{aa}, d_{ab}, d_{bb}$ can be estimated via maximum likelihood using the \texttt{mvmeta} package for multivariate meta-regression, and we denote these as $\widehat \beta_S,\widehat \beta_Y, \widehat d_{aa}, \widehat d_{ab}, \widehat d_{bb}$. 

\cite{elliott2015surrogacy} define a number of useful quantities relating to the surrogate paradox; however, the quantity that is most comparable to our defined resilience probability is:
\begin{eqnarray*}
p_e &=& P\left (\Delta_{Y(K+1)} < 0 \bigg | \Delta_{S(K+1)} \right ) \\
&=& \Phi\left(
0;\,
\beta_Y + \frac{d_{ab}}{d_{aa}}\left(\Delta_{S(K+1)} - \beta_S\right),\,
d_{bb} - \frac{d_{ab}^2}{d_{aa}} \right )
\end{eqnarray*}
where $\Phi\left (0; \mu, \sigma^2  \right )$ denotes the cumulative distribution function of a normal random variable with mean $\mu$ and variance $\sigma^2$. Thus, we estimate this quantity as:
\begin{eqnarray*}
\widehat{p}_e &=& \Phi\left(
0;\,
\widehat{\beta}_Y + \frac{\widehat{d}_{ab}}{\widehat{d}_{aa}}\left(\widehat{\Delta}_{S(K+1)} - \widehat\beta_S\right),\,
\widehat d_{bb} - \frac{\widehat d_{ab}^2}{ \widehat d_{aa}} \right )
\end{eqnarray*}
where $\widehat{\Delta}_{S(K+1)} = \sum_{i=1}^{n_{K+1,1}} S_{i1(K+1)} -  \sum_{j=1}^{n_{K+1,0}} S_{j0(K+1)} $.  

Importantly, while some assumptions underlying the framework of \citet{elliott2015surrogacy} align with ours, the fact that we do not generate the surrogate values according to the same distribution in Study $K+1$ means that the probabilities obtained from the \citet{elliott2015surrogacy} and our procedure should not be interpreted as equivalent quantities. Nevertheless, we include this comparison as a benchmark, as it provides insight into the behavior of the estimator of \citet{elliott2015surrogacy} under a data-generating mechanism that departs from its underlying assumptions, and because it represents the closest available comparison to existing methods.
\subsection*{Additional Simulation Results}

 Simulation results for all settings where $K=10$ and $n=100$ using the PAB procedure are displayed in Table \ref{sim_results_PAB}. Overall, the PAB procedure performs reasonably well across most settings, though it tends to produce confidence intervals with below-nominal coverage in several cases, most notably in Settings 5 and 6 under linear mean specification (where the linear mean specification is incorrect). This is consistent with the limitations of the PAB approach, stated in the main text: it relies on several approximations and requires correct specification of the mean model, making it more sensitive to model misspecification than the nonparametric bootstrap. Accordingly, the nonparametric bootstrap-based results in the main text generally exhibit better coverage and thus, though more computationally expensive, we suggest the nonparametric bootstrap in practice. 
 \clearpage
\begin{table}[H]
\centering
\small
\setlength{\tabcolsep}{3.5pt}

\begin{tabular}{l c| cccc | cccc | cccc }
\hline
\multicolumn{14}{c}{\textbf{$\mathbf{K=10}$, sample size per study and arm $\mathbf{=100}$, PAB Procedure}} \\
\hline
Setting & Truth
& \multicolumn{4}{c|}{Linear}
& \multicolumn{4}{c|}{Cubic}
& \multicolumn{4}{c}{Cubic Spline} \\
\hline
 & & Est  & ESE & ASE & CP
   & Est  & ESE & ASE & CP
   & Est & ESE & ASE & CP \\
\hline
1&0.586
&0.565&0.144&0.123&0.870
&0.575&0.148&0.133&0.879
&0.572&0.139&0.136&0.900\\
2&0.770
&0.867&0.029&0.078&0.949
&0.768&0.150&0.136&0.835
&0.768&0.162&0.136&0.810\\
3&0.507
&0.406&0.114&0.101&0.760
&0.427&0.142&0.117&0.804
&0.493&0.115&0.127&0.910\\
4&0.029
&0.032&0.036&0.026&0.700
&0.033&0.037&0.029&0.780
&0.034&0.035&0.033&0.800\\
5&0.078
&0.840&0.033&0.086&0.001
&0.087&0.082&0.072&0.807
&0.076&0.059&0.070&0.890\\
6&0.000
&0.436&0.071&0.048&0.000
&0.074&0.077&0.046&0.798
&0.012&0.015&0.013&0.910\\
\hline
\end{tabular}
\caption{Simulation results examining the bias and variance of estimated $\widehat{p}$ values for Settings 1-6, over 1000 iterations, with a sample size of 10 patients in each study in each group, and 25 studies. The first column represents the true probability of the paradox in the new study. The Est column is the average and the ESE column is the empirical standard deviation of the $\widehat{p}$ values obtained using our proposed method, using linear, cubic, or cubic splines, summarized over 1000 simulation iterations. The ASE column indicates the average standard error and CP indicates the coverage of the 95\% confidence intervals using the \textbf{PAB  procedure}. Recall that in Settings 1 and 4, the conditional functions are truly linear; in Settings 2 and 5, the conditional functions are linear/quadratic; in Settings 3 and 6, the conditional functions are neither linear or cubic. } \label{sim_results_PAB}
\end{table}

\clearpage
\bibliographystyle{plainnat}
\bibliography{Surrogate_master}

\end{document}